\documentclass{article}
\usepackage[utf8]{inputenc} %
\usepackage[T1]{fontenc}    %
\usepackage{tgpagella}

\usepackage{graphicx} %
\usepackage{amsmath}
\usepackage{enumitem}
\usepackage[super]{nth}

\newcommand{\CC}{\mathsf{CC}}
\newcommand{\cut}{\mathsf{CUT}}
\newcommand{\density}{\mathbf{density}}
\newcommand{\cE}{\mathcal{E}}
\newcommand{\IS}{\mathsf{IS}}
\newcommand{\add}{\mathsf{ADD}}

\usepackage{array}
\usepackage{tabularx}
\usepackage{graphicx}
\usepackage{amsmath,amssymb,amsthm,mathtools}
\usepackage{bm}
\usepackage{bbm}
\usepackage{xspace}
\usepackage{url}
\usepackage{fullpage, prettyref}
\usepackage{boxedminipage}
\usepackage{wrapfig}
\usepackage{ifthen}
\usepackage{color}
\usepackage{xcolor}
\usepackage{framed}

\usepackage[ruled]{algorithm2e}
\usepackage{algpseudocode}

\usepackage[pagebackref,colorlinks=true,pdfpagemode=UseNone,urlcolor=blue,linkcolor=cyan,citecolor=violet,pdfstartview=FitH]{hyperref}
\usepackage[capitalise,noabbrev, nameinlink]{cleveref}
\usepackage{fullpage}
\usepackage{esvect}
\usepackage{mdframed}

\usepackage{thmtools}
\usepackage{thm-restate}

\newtheorem{theorem}{Theorem}[section]

\newtheorem{lemma}[theorem]{Lemma}
\newtheorem{claim}[theorem]{Claim}
\newtheorem{corollary}[theorem]{Corollary}

\newtheorem{fact}[theorem]{Fact}

\newtheorem{question}[theorem]{Question}

\newcommand{\ignore}[1]{}

\newcommand{\cA}{{\cal A}}

\newcommand{\cD}{\mathcal{D}}

\newcommand{\cG}{\mathcal{G}}

\newcommand{\eps}{\varepsilon}

\newcommand{\poly}{\mathrm{poly}}

\newcommand{\ceil}[1]{\lceil#1\rceil}

\newcommand{\EX}{\mathbf{E}}
\newcommand{\prob}{\mathbf{Pr}}

\newcommand{\Sec}[1]{\hyperref[sec:#1]{\Cref*{sec:#1}}} %
\newcommand{\Eqn}[1]{\hyperref[eq:#1]{(\ref*{eq:#1})}} %
\newcommand{\Fig}[1]{\hyperref[fig:#1]{Fig.\,\ref*{fig:#1}}} %
\newcommand{\Tab}[1]{\hyperref[tab:#1]{Tab.\,\ref*{tab:#1}}} %
\newcommand{\Thm}[1]{\hyperref[thm:#1]{Theorem\,\ref*{thm:#1}}} %
\newcommand{\Fact}[1]{\hyperref[fact:#1]{Fact\,\ref*{fact:#1}}} %
\newcommand{\Lem}[1]{\hyperref[lem:#1]{Lemma\,\ref*{lem:#1}}} %
\newcommand{\Prop}[1]{\hyperref[prop:#1]{Prop.~\ref*{prop:#1}}} %
\newcommand{\Cor}[1]{\hyperref[cor:#1]{Corollary~\ref*{cor:#1}}} %
\newcommand{\Conj}[1]{\hyperref[conj:#1]{Conjecture~\ref*{conj:#1}}} %
\newcommand{\Def}[1]{\hyperref[def:#1]{Definition~\ref*{def:#1}}} %
\newcommand{\Alg}[1]{\hyperref[alg:#1]{Alg.~\ref*{alg:#1}}} %
\newcommand{\Obs}[1]{\hyperref[obs:#1]{Obs.~\ref*{obs:#1}}} %
\newcommand{\Ex}[1]{\hyperref[ex:#1]{Ex.~\ref*{ex:#1}}} %
\newcommand{\Clm}[1]{\hyperref[clm:#1]{Claim~\ref*{clm:#1}}} %
\newcommand{\Step}[1]{\hyperref[step:#1]{Step~\ref*{step:#1}}} %

\crefname{algocf}{alg.}{algs.}
\Crefname{algocf}{Alg.}{Algs.}

\title{Optimal Graph Reconstruction by Counting Connected Components in Induced Subgraphs}
\author{
Hadley Black\thanks{All authors supported by NSF HDR TRIPODS Phase II grant 2217058 (EnCORE Institute).} \\
UC San Diego 
\and Arya Mazumdar\footnotemark[1] \\
UC San Diego
\and Barna Saha\footnotemark[1] \\
UC San Diego 
\and  Yinzhan Xu\footnotemark[1] \\
UC San Diego
}
\date{}

\begin{document}

\maketitle

\begin{abstract}

The graph reconstruction problem has been extensively studied under various query models. In this paper, we propose a new query model regarding the number of connected components, which is one of the most basic and fundamental graph parameters. Formally, we consider the problem of reconstructing an $n$-node $m$-edge graph with oracle queries of the following form: provided with a subset of vertices, the oracle returns the number of connected components in the induced subgraph. We show $\Theta(\frac{m \log n}{\log m})$ queries in expectation are both sufficient and necessary to adaptively reconstruct the graph. In contrast, we show that $\Omega(n^2)$ non-adaptive queries are required, even when $m = O(n)$. We also provide an $O(m\log n + n\log^2 n)$ query algorithm using only two rounds of adaptivity.

\end{abstract}

\section{Introduction}

Graph reconstruction is a classical problem that aims to learn a hidden graph by asking queries to an oracle. Over the last three decades, the problem has been extensively studied under various query models (e.g., \cite{DBLP:conf/cocoon/Grebinski98, DBLP:journals/dam/GrebinskiK98, DBLP:journals/siamcomp/AlonBKRS04, DBLP:journals/siamdm/AlonA05, DBLP:conf/alt/ReyzinS07, DBLP:journals/jcss/AngluinC08, CK10, DBLP:conf/soda/Mazzawi10,DBLP:journals/tcs/BshoutyM11,DBLP:journals/tcs/BshoutyM12, DBLP:conf/colt/Choi13, DBLP:journals/siamdm/BshoutyM15,KST25}). 
Given a graph $G=(V,E)$ with an unknown edge set $E$, an algorithm may submit queries to an oracle (what the oracle computes is problem-specific) in order to recover $E$. The main objective is to minimize the number of queries.

Among all types of queries, one of the simplest and oldest is the \emph{edge-detection} model (also referred to as \emph{independent set} ($\mathsf{IS}$) queries), which was studied since \cite{DBLP:journals/dam/GrebinskiK98}. In this model, the algorithm may submit a set $S \subseteq V$ to the oracle, which returns whether the induced subgraph $G[S]$ contains at least one edge. At the beginning, this model was studied for special classes of graphs such as Hamiltonian cycles \cite{DBLP:journals/dam/GrebinskiK98}, matchings \cite{DBLP:journals/siamcomp/AlonBKRS04}, stars, and cliques \cite{DBLP:journals/siamdm/AlonA05}. The problem was then studied on general graphs by \cite{DBLP:journals/siamdm/AlonA05, DBLP:conf/alt/ReyzinS07} 
and the best query complexity known is $O(m \log n)$ by \cite{DBLP:journals/jcss/AngluinC08}, which is optimal when $m \le n^{2-\eps}$ for any $\eps > 0$. 

Several natural ways to strengthen the edge-detection query model have also been studied: 

\begin{itemize}\itemsep0em 
    \item \emph{Additive ($\add$) queries:} Given set $S \subseteq V$, the oracle returns the number of edges in $G[S]$ \cite{DBLP:conf/cocoon/Grebinski98, DBLP:journals/algorithmica/GrebinskiK00, DBLP:journals/tcs/BshoutyM11, CK10,DBLP:conf/soda/Mazzawi10,  DBLP:journals/tcs/BshoutyM12, DBLP:conf/colt/Choi13, DBLP:journals/siamdm/BshoutyM15}.
    \item \emph{Cross-additive ($\cut$) queries:} Given disjoint sets $S,T \subseteq V$, the oracle returns the number of edges between $S$ and $T$ \cite{CK10, DBLP:conf/colt/Choi13}.
    \item \emph{Maximal independent set ($\mathsf{MIS}$) queries:} Given set $S \subseteq V$, the oracle returns an adversarially chosen maximal independent set in $G[S]$ \cite{KST25}.
\end{itemize}

The additive and cross-additive models have garnered significant attention %
and are known to have query complexity $\Theta(\frac{m \log (n^2/m)}{\log m})$ \cite{CK10, DBLP:conf/soda/Mazzawi10}. For weighted graphs, where instead of the number of edges, the oracle returns the total weights of these edges and the goal is the to reconstruct all edges with their weights, the query complexity becomes $\Theta(\frac{m \log n}{\log m})$ \cite{DBLP:journals/tcs/BshoutyM11, DBLP:conf/colt/Choi13, DBLP:journals/siamdm/BshoutyM15}. Interestingly, there exist non-adaptive algorithms  that attain these query complexities in both models and for both unweighted and weighted graphs.\footnote{An algorithm is called non-adaptive if all of its queries can be specified in one round, and adaptive otherwise.} These models also have connections with classic algorithmic and combinatorial questions such as \emph{group testing} and \emph{coin-weighing}. 

Maximal independent set queries were introduced very recently by \cite{KST25}, motivated by the existence of efficient algorithms for maximal independent set in various computation models such as the Congested-Clique model, the LOCAL model and the semi-streaming model. These efficient algorithms make maximal independent set queries a potential candidate to use as a building block in more complicated algorithms. \cite{KST25} was also motivated to study whether $\mathsf{MIS}$-queries are strictly more powerful than $\IS$-queries. 
They obtained upper and lower bounds on the query complexity mainly in terms of the maximum degree of the graph.

\paragraph{Connected Component Count Queries.}

In this paper, we initiate the study of graph reconstruction using an oracle which returns the \emph{number of connected components} in $G[S]$ ($\CC$-queries), which is one of the most basic and fundamental graph parameters. This model is a natural way to strengthen $\IS$-queries, as there exists an edge in $G[S]$ if and only if the number of connected components is strictly less than $|S|$. %

Similar to maximal independent set, the number of connected components can also be computed efficiently in other computation models. For instance, \cite{DBLP:conf/podc/GhaffariP16} showed how to compute the number of connected components in $O(\log^* n)$ rounds in the Congested-Clique model. An efficient $\tilde{O}(n)$-space streaming algorithm also trivially exists, by maintaining an arbitrary spanning forest. These efficient algorithms make it hopeful to use $\CC$-queries as a basic building block in other algorithmic applications. It is worth noting that, estimating the number of connected components was also considered in the statistics literature~\cite{frank1978estimation,bunge1993estimating}.

An additional motivation for our study of graph reconstruction with $\CC$ queries is that it generalizes the problem of learning a partition using rank/subset queries, studied recently by \cite{DBLP:conf/fsttcs/Chakrabarty024} and \cite{BLMS24}. In particular, this problem corresponds to the special case when the hidden graph $G$ is known to be a disjoint union of cliques and is itself a generalization of the well-studied problem of clustering using pair-wise same-set queries (e.g. \cite{DBLP:conf/alt/ReyzinS07,balcan2008clustering,mitzenmacher2016predicting,ashtiani2016clustering,MS17a,MS17b,mazumdar2017theoretical,saha2019correlation, huleihel2019same,bressan2020exact,liu2022tight,del2022clustering,DMT24}).

\subsection{Results}

As our main result, we settle the query complexity of graph reconstruction with $\CC$-queries, by showing that $\Theta(\frac{m\log n}{\log m})$ queries are both necessary and sufficient for any value of $m$. Our main result can be more formally described as the following two theorems:

\begin{restatable}[Adaptive Algorithm]{theorem}{AdaptiveAlgo}
\label{thm:adaptive-alg}
There is an adaptive, randomized polynomial time algorithm that, given $\CC$-query access to an underlying (unknown) $n$-node graph $G = (V, E)$, the vertex set $V$ and an upper bound $m$ on the number of edges, reconstructs $G$ using $O(\frac{m \log n}{\log m})$ $\CC$-queries in expectation. 
\end{restatable}

\begin{restatable}[Adaptive Lower Bound]{theorem}{AdaptiveLower}
    \label{thm:LB-A} Any (randomized) adaptive graph reconstruction algorithm requires $\Omega(\frac{m \log n}{\log m})$ $\CC$-queries in expectation to achieve $1/2$ success probability for all $m \leq \binom{n}{2} - 1$ and even when $m$ is provided as input.
\end{restatable}

Compared with the $O(m \log n)$ bound for $\IS$-queries \cite{DBLP:journals/jcss/AngluinC08}, which is tight when $m \le n^{2-\eps}$ for any $\eps > 0$, our bound shows that $\CC$-queries are indeed strictly stronger than $\IS$-queries, at least when $\omega(1) < m < n^{2-\eps}$. 

The comparison between our bound and the bounds for (cross)-additive queries is more intriguing. Our bound coincides with the bound of (cross)-additive queries for \emph{weighted graphs} \cite{DBLP:journals/tcs/BshoutyM11, DBLP:conf/colt/Choi13, DBLP:journals/siamdm/BshoutyM15}; whereas for unweighted graphs, the $\Theta(\frac{m \log (n^2/m)}{\log m})$ bound achieved by \cite{CK10, DBLP:conf/soda/Mazzawi10} is smaller than our bound when $m$ is very close to $n^2$. In particular, for dense graphs where $m = \Omega(n^2)$, $\Theta(n^2)$ $\CC$-queries are required, while only $\Theta(n^2/\log n)$ (cross)-additive queries suffice for unweighted graphs.

Since our bound is similar to those of (cross)-additive queries, and there are non-adaptive algorithms that achieve the same bounds for (cross)-additive queries for both unweighted graphs \cite{CK10} and for weighted graphs \cite{DBLP:journals/tcs/BshoutyM11, DBLP:journals/siamdm/BshoutyM15}, it is natural to ask whether it is possible to obtain the optimal query complexity for $\CC$-queries using a \emph{non-adaptive} algorithm. Perhaps surprisingly, we prove that this is not possible for $\CC$ queries in a strong sense: we show that non-adaptive algorithms require $\Omega(n^2)$ $\CC$-queries even for \emph{sparse} graphs, matching the trivial upper bound obtained by simply querying every pair of nodes. This establishes a stark contrast between $\CC$ queries and (cross)-additive queries.

\begin{restatable}[Non-Adaptive Lower Bound]{theorem}{NonAdaptiveLower}
    \label{thm:LB-NA} Any non-adaptive graph reconstruction algorithm requires $\Omega(n^2)$ $\CC$-queries in expectation to achieve $1/2$ success probability, even when the graph is known to have $O(n)$ edges.
\end{restatable}

\cref{thm:LB-NA} also implies $\Omega(n^2)$ $\IS$-queries are required on sparse graphs for randomized algorithms, which is on par with the $\Omega(m^2 \log n)$ non-adaptive lower bound for $\IS$-queries shown in \cite{DBLP:conf/alt/AbasiN19} for certain ranges of $m$. 

In light of this strong lower bound for non-adaptive algorithms, we investigate whether efficient algorithms are possible using few rounds of adaptivity. We show there is an algorithm using only \emph{two rounds} of adaptivity which comes within an $O(\log^2 n)$ factor of the optimal query complexity\footnote{In fact, when $m \geq \Omega(n \log n)$ our algorithm is within a factor of $O(\log n)$ of the optimal query complexity.}.

\begin{restatable}[Two Round Algorithm]{theorem}{TwoRound}
\label{thm:2round-alg}
    There is a randomized algorithm using two rounds of adaptivity and $O(m \log n + n\log^2 n)$ $\CC$-queries in the worst case which successfully reconstructs any arbitrary $n$-node graph with at most $m$ edges with probability $1-1/\poly(n)$. The upper bound $m$ is provided as input to the algorithm.
\end{restatable}

\cite{DBLP:conf/alt/AbasiN19} proved that with two rounds of adaptivity $m^{4/3-o(1)}\log n$ $\IS$-queries are required, which again separates $\IS$-queries with $\CC$-queries.

Our two round algorithm utilizes an interesting connection with group testing and can be described at a high level as follows. The first round of queries is used to obtain a constant factor approximation of the degree of every vertex using $O(n \log^2 n)$ queries. Then, once an approximation of the degree $d_v$ of a vertex $v$ is known, a simple group testing procedure can be leveraged to learn its neighborhood with $O(d_v \log n)$ queries. The second round runs this procedure in parallel for every vertex for a total of $O(m \log n)$ queries.

\paragraph{Paper Overview.} The main effort of our paper is in proving the correctness of our optimal adaptive algorithm, \Cref{thm:adaptive-alg}. Our proof uses multiple important subroutines  which we describe in \Cref{sec:useful-subroutines} and \Cref{sec:key-subroutine}. Given these subroutines, the main proof is given in \Cref{sec:a-alg-proof}. A high level overview of the proof is given in \Cref{sec:technical-overview}. The proof of correctness for our two round algorithm (\Cref{thm:2round-alg}) is given in \Cref{sec:2-round-proof} and we prove our lower bounds (\Cref{thm:LB-A} and \Cref{thm:LB-NA}) in \Cref{sec:LB}.

\subsection{Other Related Works}

Aside from the edge-detection, (cross-)additive,  and maximal independent set already discussed, the graph reconstruction problem has also been studied extensively in the \emph{distance} query model by \cite{DBLP:conf/sirocco/KranakisKU95, DBLP:journals/dam/GrebinskiK98, DBLP:journals/jsac/BeerliovaEEHHMR06, DBLP:conf/ciac/ErlebachHHM06, DBLP:conf/alt/ReyzinS07,
DBLP:conf/icalp/MathieuZ13, DBLP:conf/icalp/KannanM015, DBLP:conf/aaim/KranakisKL16,DBLP:journals/talg/KannanMZ18,DBLP:journals/rsa/MathieuZ23,DBLP:journals/corr/abs-2306-05979}. In this model, the algorithm may query a pair of vertices $(x,y)$ and the oracle responds with the shortest path distance from $x$ to $y$. This model has a notably different flavor from edge-detection and its related models. For bounded degree graphs, the current best known algorithm achieves query complexity $\widetilde{O}(n^{3/2})$, and it is a significant open question as to whether $\widetilde{O}(n)$ can be achieved. 

Query algorithms have also received a great deal of attention recently for solving various (non-reconstruction) problems over graphs, for example testing connectivity, computing the minimum cut, and computing spanning forests \cite{harvey2008matchings, DBLP:journals/corr/abs-2109-02115, DBLP:conf/esa/AssadiCK21, DBLP:conf/focs/ApersEGLMN22,DBLP:conf/alt/ChakrabartyL23,DBLP:conf/alt/0001C24,DBLP:conf/soda/0001SW25}. $\mathsf{AND}$, $\mathsf{OR}$, and $\mathsf{XOR}$ query-algorithms have also been considered for the problem of computing maximum cardinality matchings in bipartite graphs (see e.g. \cite{DBLP:conf/focs/BlikstadBEMN22}), which is a question with particular importance in communication complexity.  

Aside from graphs, there is a rich body of literature concerning reconstruction of combinatorial objects, such as strings (e.g., \cite{DBLP:conf/colt/HoldenPP18, DBLP:conf/stoc/000121}), partitions (e.g., \cite{DBLP:conf/fsttcs/Chakrabarty024,BLMS24}), matrices (e.g., \cite{DBLP:journals/siamcomp/BoutsidisDM14,DBLP:conf/stoc/CohenEMMP15}), and more.

\section{Preliminaries}

\label{sec:prelim}

Throughout this paper, we consider unweighted undirected graphs $G=(V, E)$. For $v \in V$, we let $\deg_G(v)$ denote its degree and let $N_G(v)$ denote the set of $v$'s neighbors. For $U \subseteq V$, we use $\deg_G(v, U)$ to denote the number of edges between $v$ and $U$. Furthermore, we use $\CC_G(U)$ to denote the number of connected components in the induced subgraph $G[U]$, and we use $\density_G(U)$ to denote the number of edges in the induced subgraph $G[U]$. All subscripts $G$ might be dropped if clear from context.

\subsection{Useful Subroutines} \label{sec:useful-subroutines}

Mazzawi \cite{DBLP:conf/soda/Mazzawi10} showed that one can use $O(\frac{m \log (n^2/m)}{\log m})$ additive queries to reconstruct an $n$-node $m$-edge graph. This result is a key subroutine of our result. For our purpose, an $O(\frac{m \log n}{\log m})$ bound already suffices, so we will utilize the latter bound for simplicity. 

\begin{theorem}[\cite{DBLP:conf/soda/Mazzawi10}] \label{thm:GR-density}
There is a polynomial-time algorithm  that, given $\density$-query access to an underlying (unknown) $n$-node graph $G = (V, E)$, the vertex set $V$ and the number of edges $m$, reconstructs the graph $G$ using $\frac{Cm \log n}{\log m}$ queries in the worst case for some universal constant $C$. 
\end{theorem}

In the following, we describe several subroutines that will be used by our algorithm. 

In some places, we will consider a more general version of graph reconstruction, where a subset of the edges are already known to the algorithm. In many scenarios, one can design algorithms whose  query complexity mostly depends on the number of unknown edges, instead of the total number of edges.

The following lemma shows a procedure to reconstruct forests. 
\begin{lemma} [Reconstructing Forests]
\label[lemma]{lem:GR-forest}
There is a $\CC$-query algorithm that, given $\CC$-query access to an underlying (unknown) $n$-node graph $G = (V, E)$, the vertex set $V$ and a set of  edges $K$, reconstructs the graph $G$ with $K \subseteq E$ when $G$ is a forest, using $O(\frac{\max(2, m_u) \log n}{ \log \max(2, m_u)})$ queries in the worst case for $m_u = |E| - |K|$. In addition, even if $G$ is not a forest, the algorithm still uses $O(\frac{\max(2, m_u) \log n}{ \log \max(2, m_u)})$ queries in the worst case, and only outputs edges in $E \setminus K$. 
\end{lemma}
\begin{proof}
At the beginning of the algorithm, we make one query to compute the number of connected components $\CC(G)$ of the whole graph, and set $\hat{m}_u = n - \CC(G) - |K| \le m_u$.

First, we assume $G$ is a forest. In this case, note that $\hat{m}_u = m_u$. We run \Cref{thm:GR-density} on the graph $G_u = (V, E \setminus K)$ and halt the algorithm once the number of queries exceeds $\frac{C\hat{m}_u \log n}{\log \hat{m}_u}$. We use $\CC$ queries on $G$ to simulate $\density$ queries on $G_u$. Whenever \Cref{thm:GR-density} queries the number of edges in an induced subgraph $G_u[U]$ for some $U \subseteq V$, we make a $\CC$ query on $G[U]$. Since $G$ is a forest, the number of edges in $G[U]$ equals $|U| - \CC(U)$. Then we can compute the number of edges in $G_u[U]$ by subtracting $|E(G[U]) \cap K|$ from the previous count. Hence, one $\CC$ query on $G$ suffices to simulate one $\density$ query on $G_u$, so the first part of the proposition follows from \Cref{thm:GR-density}. 

Since we always halt the algorithm when \cref{thm:GR-density} makes more than $\frac{C\hat{m}_u \log n}{\log \hat{m}_u} \le \frac{Cm_u \log n}{\log m_u}$ queries, the number of queries is bounded by $O(1+\frac{m_u \log n}{ \log m_u}) \le O(\frac{\max(2, m_u) \log n}{ \log \max(2, m_u)})$, even when $G$ is not a forest. Furthermore, as we can verify in one query whether an edge output by the algorithm actually belongs to $E \setminus K$, the second part of the lemma also follows. Formally, we can iterate over the candidate edges returned by the algorithm, spending one query to verify each and halting as soon as a pair not belonging to $E \setminus K$ is seen. Thus we spend at most $m_u$ queries for the verification and since $\log(m_u) = O(\log n)$, we have $m_u = O(\frac{\max(2, m_u) \log n}{ \log \max(2, m_u)})$, and so this does not dominate the query complexity. \end{proof}

Next, we show a simple algorithm that reconstructs the whoel graph using binary search, but  uses up to poly-logarithmically more queries than the optimal algorithm. 

\begin{lemma} [Reconstructing with Binary Search]
\label[lemma]{lem:GR-binary-search}
There is an algorithm that, given $\CC$-query access to an underlying (unknown) $n$-node graph $G = (V, E)$, the vertex set $V$ and a set of  edges $K$, reconstructs the graph $G$ with $K \subseteq E$ using $O((n+m_u) \cdot \log (\frac{n^2}{n + m_u}))$ queries in the worst case, where $m_u = |E| - |K|$.
\end{lemma}
\begin{proof}
    For each vertex $v \in V$, let $U_v$ be the set of vertices in $V \setminus \{v\}$ that do not yet have an edge in $K$ with $v$. Note that for any subset $U \subseteq V \setminus \{v\}$, there is at least one edge between $v$ and $U$ if and only if $\CC(U) = \CC(U \cup \{v\})$. Therefore, we can perform a binary-search style algorithm to find all neighbors of $v$ in $U_v$. More specifically, we use the above method to check whether $v$ has at least one edge with $U_v$. If not, we halt and report no edge. Otherwise, we split $U_v$ into two equal halves (up to $1$ vertex), and recursively solve the problem in the two halves (unless $|U_v| = 1$, in which case we can simply report this edge). Let $d = \max(1, \deg(v, U_v))$. The total number of subsets we query with size at least $n / d$ is $O(d)$, as there are only $O(d)$ of them in the recursion process; for smaller subsets, the binary search costs $O(\log(n/d) + 1) = O(\log(2d/d))$ queries. Hence, the overall number of queries needed is $O(d \log(2n/d)) = O\left(\max(1, \deg(v, U_v)) \cdot \log \left(\frac{2n}{\max(1, \deg(v, U_v))}\right)\right)$.

    Summing over all vertices, the total number of queries is 
    \[
    O\left(\sum_{v \in V} \max(1, \deg(v, U_v)) \cdot \log \left(\frac{2n}{\max(1, \deg(v, U_v))}\right) \right). 
    \]
    By Jensen's inequality, the above can be upper bounded by 
    \[
    O\left(n\cdot \left(1+\frac{2m_u}{n}\right) \cdot \log \left(\frac{2n}{1+\frac{2m_u}{n}}\right)  \right) =  O\left((n+m_u) \cdot \log \left( \frac{n^2}{n + m_u} \right)\right).
    \]
\end{proof}

The next lemma shows an algorithm that finds the neighbors of some vertex $v$ in a subset of vertices $U$, when all edges in the induced subgraphs of $U$ are already known. 

\begin{lemma}
\label[lemma]{lem:GR-find-neighbors-IS}
There is an algorithm that, given $\CC$-query access to an underlying (unknown) $n$-node graph $G = (V, E)$, the vertex set $V$, a subset of vertices $U \subseteq V$, all edges between vertices in $U$, an upper bound $D$ on the maximum degree of the induced subgraph on $U$, and a vertex $v \in V \setminus U$, finds all edges between $v$ and $U$ using 
\[
O\left((D + \deg(v, U)) \cdot  \frac{\log (|U| / D + 1)}{ \log(\max(2, \deg(v, U) / D))}\right)\]
$\CC$-queries in the worst case. 
\end{lemma}
\begin{proof}
Since the maximum degree of $G[U]$ is at most $D$, we can color $G[U]$ using $D+1$ colors. Then each color class forms an independent set and we denote the independent sets by $I_1, \ldots, I_{D+1}$. If some independent set has size greater than $|U| / D$, we arbitrarily decompose the independent set to smaller ones of size $O(|U| / D)$. The total number of independent sets is still $\ell = O(D)$. %
Between every independent set $I_i \subseteq U$ and $v$, we run \Cref{lem:GR-forest} on the subgraph induced on $I_i \cup \{v\}$. As $I_i$ is an independent set, this induced subgraph is a forest, so \Cref{lem:GR-forest} finds all edges between $v$ and $I_i$ in 
\[
O\left(\frac{\max(2, \deg(v, I_i)) \cdot \log (|U| / D + 1)}{\log \max(2, \deg(v, I_i))} \right)
\]
queries. By Jensen's inequality, the total number of queries is hence %
\begin{align*}
    & O\left(\sum_{i=1}^{\ell} \frac{\max(2, \deg(v, I_i)) \log (|U|/ D + 1)}{\log \max(2, \deg(v, I_i))}\right).\\
\end{align*}
By upper bounding $\max(2, \deg(v, I_i))$ by $e^2+\deg(v, I_i)$ and use the fact that the function $x/\log x$ is concave for $x > e^2$ and apply Jensen's inequality, the above can be bounded by
\begin{align*}
& O\left(\ell \cdot \left( e^2+\deg(v, U) / \ell\right) \cdot \frac{\log (|U|/D + 1)}{\log(e^2+\deg(v, U) / \ell)}\right) \\
     & = O\left((D+ \deg(v, U)) \cdot \frac{\log (|U| / D + 1)}{\log(\max(2, \deg(v, U) / D)))}\right).\\
\end{align*}
\end{proof}

The next lemma shows an algorithm that finds all vertices adjacent to a subset of vertices $S$. 

\begin{lemma}
\label[lemma]{lem:find-adjacent-to-S}
There is an algorithm that, given $\CC$-query access to an underlying (unknown) $n$-node graph $G = (V, E)$, the vertex set $V$ and a subset of vertices $S \subseteq V$, outputs all vertices in $V$ with at least one edge connected to some vertex in $S$, using $O((\sum_{s \in S} \deg(s)) \log n)$ $\CC$-queries in the worst case. 
\end{lemma}
\begin{proof}
    Given an arbitrary subset $U \subseteq V$, we can determine whether there is an edge between $S$ and $U$ using $O(1)$ $\CC$ queries as follows. Let $W = S \cap U$. Then there is no edge between $S$ and $U$ if and only if 
    \[
        \CC(S \setminus W) + \CC(W) + \CC(U \setminus W) = \CC(S \cup U) \text{    and    } \CC(W) = |W|. 
    \]
    
    Therefore, we can find all vertices adjacent to $S$ using a binary-search approach. Start by setting $U = V$. If we determine (using $O(1)$ queries as described in the previous paragraph) that there is no edge between $S$ and $U$, we do not need to do any work. Otherwise, if $|U| = 1$, we output the vertex in $U$ as one vertex adjacent to $S$. Otherwise, we arbitrary split $U$ to two halves $U_1$ and $U_2$, and recursively find all neighbors of $S$ in $U_1$ and $U_2$. It is not difficult to bound the number of queries of this procedure by $O((\sum_{s \in S} \deg(s)) \log n)$. 
\end{proof}

The final lemma in this section shows an algorithm that finds all vertices whose degrees are approximately equal to some given parameter $T$. 
\begin{lemma}
\label[lemma]{lem:find-high-degree-nodes}
There is an algorithm that, given $\CC$-query access to an underlying (unknown) $n$-node graph $G = (V, E)$, the vertex set $V$, an upper bound $m$ on the number of edges, an integer parameter $T \ge 10$, and a parameter $0 < \delta < 1$, outputs a set $H \subseteq V$ using $O(\frac{m (\log m + \log \frac{1}{\delta}) \log n}{T})$ $\CC$-queries in the worst case. With probability $\ge 1 - \delta$, 
\begin{itemize}
    \item for every $v \in V$ with $\deg(v) \ge 2T$, $v \in H$;
    \item for every $v \in V$ with $\deg(v) \le T / 2$, $v \not \in H$. 
\end{itemize}
\end{lemma}
\begin{proof}
The algorithm is described in \Cref{alg:find-high-degree-nodes}. 
\begin{algorithm}[ht]
\caption{Algorithm for finding high-degree vertices} \label{alg:find-high-degree-nodes}
\KwIn{$V, m, T, \delta$ and $\CC$-query access to a graph $G=(V, E)$ with $n$ nodes and at most $m$ edges.}
\KwOut{A subset $H \subseteq V$ with the property described in \cref{lem:find-high-degree-nodes}.} 
$\ell \gets \lceil 200(\ln(2m) + \ln(1/\delta)\rceil$\; 
Initialize an array $a$ indexed on $V$ initially all $0$\; 
\For{$i = 1$ to $\ell$}{
Sample a subset $S \subseteq V$ where each $v \in V$ is added to $S$ independently with probability $\frac{1}{T}$\; 
Run \Cref{lem:find-adjacent-to-S} with $V, S$ to get a set $U \subseteq V$ but halt \Cref{lem:find-adjacent-to-S} when the number of queries exceeds $O(\frac{20m}{T} \log n)$ where the constant factor hidden in $O$ is the same as the constant factor in \Cref{lem:find-adjacent-to-S}\; 
\If{\Cref{lem:find-adjacent-to-S} finishes}
{
\For{$v \in U$}
{
$a_v \gets a_v + 1$\;
}
}
}
\Return $H = \{v \in V: a_v \ge \ell / 2\}$\;
\end{algorithm}

First, we analyze the number of queries used by \Cref{alg:find-high-degree-nodes}. Each iteration of the for-loop uses $O(\frac{m \log n}{T})$ queries, so the overall number of queries is $O(\frac{m \ell \log n}{T}) = O(\frac{m (\log m + \log \frac{1}{\delta}) \log n}{T})$.  

Next, we analyze the error probability. For every $v \in V$ with $\deg(v) \ge 2T$ and for some iteration $i$ of the algorithm, let $\cE_1$ be the event that the sample $S$ is disjoint from the neighbors of $v$. We have that 
\[
\prob[\cE_1] \le  \left(1-\frac{1}{T}\right)^{2T} \le e^{-2} \le 0.2. 
\]
Furthermore, let $\cE_2$ be the event that $\sum_{s \in S} \deg(s) \ge \frac{20 m}{T}$. It is not difficult to see that 
\[
\EX\left[\sum_{s \in S} \deg(s)\right] \le \frac{2m}{T}. 
\]
Therefore, by Markov's inequality, $\prob[\cE_2] \le 0.1$. As long as neither $\cE_1$ nor $\cE_2$ happens, $v$ will be among the vertices returned by \Cref{lem:find-adjacent-to-S}, and thus $a_v$ will be incremented in the iteration. Hence, $a_v$ is incremented in each iteration with probability $\ge 1 - \prob[\cE_1 \vee \cE_2] \ge 0.7$. By Chernoff bound, at the end of the algorithm,
\[
\prob\left[a_v \le \ell / 2\right] \le \prob\left[a_v \le (1 - 0.2 / 0.7) \EX[a_v]\right] \le \exp(-(0.2 / 0.7)^2 \cdot (0.7 \ell) / 2) \le \frac{\delta}{2m}. 
\]
In other words, $v$ will be output by the algorithm with probability $\ge 1 - \frac{\delta}{2m}$. 

Next, we consider vertices $v \in V$ with $\deg(v) \le T / 2$. For any iteration $i$, the probability that $S$ is disjoint from the neighbors of $v$ is (recall $T \ge 10$)
\[
\ge \left(1-\frac{1}{T} \right)^{T/2} \ge \left(1-\frac{1}{10} \right)^{10/2} \ge 0.59. 
\]
Therefore, the probability that $a_v$ is incremented in the iteration is $\le 0.41$. Thus, by Chernoff bound, at the end of the algorithm, 
\[
\prob\left[a_v \ge \ell / 2\right] \le \exp(-(0.09 / 0.41)^2 \cdot (0.41 \ell) / (2 + (0.09 / 0.41))) \le \frac{\delta}{2m}. 
\]

Furthermore, notice that the algorithm never outputs a vertex $v$ with degree $0$, since $a_v$ is always $0$ for such vertices. 

Finally, the lemma follows by union bound over all vertices whose degree is nonzero (there can be at most $2m$ such vertices). 
\end{proof}

\section{Adaptive Algorithm}

\subsection{Technical Overview for the Adaptive Algorithm} \label{sec:technical-overview}
In this section, we describe the high-level ideas for our optimal adaptive algorithm. For simplicity, we only consider the case $m \ge n$ in the overview, so that the number of queries we aim for is $O(\frac{m \log n}{\log m}) = O(m)$. 

Our key observation is that, when the input graph is a forest, we  could use one $\CC$ query to efficiently simulate one $\density$ query. In fact, when the input graph $G=(V, E)$ is a forest, any induced subgraph $G[U]$ for $U \subseteq V$ is also a forest, so we have $\density(U) = |U| - \CC(U)$. Hence, we could use $O(m)$ $\CC$-queries to simulate \cref{thm:GR-density} for graph reconstruction which uses $O(\frac{m \log n}{\log m}) = O(m)$ $\density$ queries. 

Hence, a natural next step is to utilize this efficient algorithm on forests, and generalize it to arbitrary graphs. We develop two main methods in order to achieve this goal: 

\paragraph{Vertices with similar degrees. } First, let us make a simplifying assumption that all vertices have similar degrees, i.e., all vertices have degrees $O(D)$ where $D = m / n$.\footnote{We set the maximum degree to be $O(m/n)$ for simplicity in this overview. In general, our algorithm also works for certain larger maximum degree. } A natural idea to obtain forests from general graphs is to sample random subgraphs. Let $0 < p < 1$ be a sample probability whose value will be determined later. Suppose we sample a random subset $S \subseteq V$ where every vertex $v \in V$ is added to $S$ with probability $p$, what is the probability that $G[S]$ is a forest? To lower bound this probability, we instead consider the probability that $G[S]$ is a matching, i.e., the maximum degree of $G[S]$ is at most $1$. Consider any vertex $u \in V$ and two of its neighbors $v, w$. The probability that they are all added in $S$ is $p^3$. The total number of such pairs is $O(mD)$, because there are $O(m)$ pairs $(u, v)$, and for fixed $u$, it can have at most $O(D)$ neighbors $w$. Hence, by union bound, the probability that any vertex in $G[S]$ has degree at least $2$ is $O(p^3 m D)$. By setting $p=\Theta((mD)^{-1/3})$ with small enough constant, we obtain that $G[S]$ is a forest with constant probability. 

However, it is not enough that $G[S]$ is a forest. Recall that the number of queries we spend on a forest with $n'$ vertices and $m'$ edges is $O(\frac{m' \log n'}{\log m'})$, which is only linear in the number of edges when $m' = (n')^{\Omega(1)}$. The expected number of vertices in $G[S]$ is $np$, and the expected number of edges is $mp^2$. Hence, we intuitively need $mp^2 = (np)^{\Omega(1)}$. By plugging in $p=\Theta((mD)^{-1/3})$ and $m = \Theta(nD)$, we get $mp^2 = \Theta((n/D)^{1/3})$ and $np = \Theta(n/D)^{2/3})$, so we indeed have $mp^2 = (np)^{\Omega(1)}$ as desired. 

The above random sampling idea leads to a procedure that recovers $\Theta(mp^2)$ edges in expectation, using $O(mp^2)$ queries. When a large fraction of edges in the graphs are discovered, we also expect that a large fraction of edges in $G[S]$ are discovered. Hence, we use a version of the algorithm for forests that does not waste queries on these already discovered edges. Eventually, the expected number of undiscovered edges in the sampled subgraph $G[S]$ might go below $(np)^{\Omega(1)}$ (for instance, when there is only one undiscovered edge in the graph), so the query complexity might become superlinear in the number of undiscovered edges once when we have discovered many edges. Despite this, we can still show that the overall expected number of queries is $O(m)$.

\paragraph{Vertices with different degrees. } To complement the above algorithm, we design the following algorithm that works when a vertex has degree much larger than some of its neighbors. Suppose $U \subseteq V$ is a subset where the maximum degrees of vertices in $U$ is $d$ and suppose we have found all edges in $G[U]$, and let $v \in V \setminus U$ be a vertex whose degree is $D \gg d$, say $D \ge d \cdot n^{0.01}$. Our algorithm will find all neighbors of $v$ in $U$ using $O(D)$ $\CC$-queries. 

Because the maximum degree of $G[U]$ is $d$, we can color the vertices of it using $d+1$ colors, which further means that we can partition $U$ into $d+1$ sets $W_1, \ldots, W_{d+1}$ so that each set is an independent set. Thus, $G[W_i \cup \{v\}]$ is a forest for every $1 \le i \le d + 1$, and we can apply the $CC$-query algorithm for a forest on it. The total query complexity is thus 
\[
O\left(\sum_{i=1}^{d+1} \deg(v, W_i) \cdot \frac{\log (|W_i + 1|^2 / \deg(v, W_i))}{\log \deg(v, W_i)}\right) = O\left(\log n \cdot \sum_{i=1}^{d+1} \frac{\deg(v, W_i)}{\log \deg(v, W_i)}\right). 
\]
Assume for simplicity that $\deg(v, W_i) \ge e^2$ for every $i$, and use the concavity of $x / \log x$ for $x > e^2$, the above can be bounded by the following using Jensen's inequality:
\[
O\left(\log n \cdot (d+1) \cdot \frac{D / (d+1)}{\log (D / (d+1))}\right) = O\left(\log n \cdot \frac{D}{\log (n^{0.01})}\right) = O(D), 
\]
as desired. 

\paragraph{Final algorithm. } The final algorithm is an intricate combination of the above two ideas. As a very high-level intuition, we carefully set up several thresholds that partition vertices to $S_1, \ldots, S_\ell$ based on their degrees. Then we use the first algorithm that works for vertices with similar degrees to reconstruct edges in $G[S_i \cup S_{i+1}]$ for $1 \le i < \ell$, and we use the second algorithm that works for vertices with very different degrees to reconstruct edges in between $S_i$ and $S_j$ for $i < j - 1$.

\subsection{A Key Subroutine}
\label{sec:key-subroutine}

In this section, we describe a key subroutine that will be used in the adaptive algorithm. 

\begin{lemma}
\label[lemma]{lem:key-subroutine}
    There is a $\CC$-query algorithm that, given $\CC$-query access to an underlying (unknown) $n$-node graph $G = (V, E)$, the vertex set $V$, an upper bound $m$ on the number of edges, an upper bound $D$ on the maximum degree, and an integer parameter $\ell \ge 0$, returns a subset $K$ of all edges $E$  so that (in the following, $p = \frac{1}{10m^{1/3}D^{1/3}}$)
    \begin{enumerate}
        \item The expected number of remaining edges $|E \setminus K|$ is upper bounded by $|E| (1-p/2)^\ell$. 
        \item The expected number of $\CC$-queries used by the algorithm is 
        \[
        O\left(\left(\frac{|E|}{\log \max(2, p^2|E|)} + \ell\right)\cdot \log(100 pn)\right) \text{.}
        \]
    \end{enumerate}
\end{lemma}

The main strategy for the proof of \cref{lem:key-subroutine} is to randomly sample induced subgraphs where each vertex is selected into the subgraph with probability $p$. 
This probability $p = \frac{1}{10m^{1/3}D^{1/3}}$ is chosen so that each sampled induced subgraph is a forest (more specifically, a matching) with good probability, and hence we can use \cref{lem:GR-forest} to reconstruct edges in the subgraph. 

The algorithm is outlined in \Cref{alg:key-subroutine}.

\begin{algorithm}[h]
\caption{A key subroutine} \label{alg:key-subroutine}
\KwIn{(1) $\CC$-query access to graph $G=(V,E)$ with $n$ vertices, at most $m$ edges and maximum degree at most $D$. (2) $V, m, D$ and a parameter $\ell \ge 0$.}
\KwOut{A set of edges $K \subseteq E$. }
$p \gets \frac{1}{10m^{1/3}D^{1/3}}$\;
$K \gets \emptyset$\;
\For{$ t = 1 \to \ell$}{
\label{line:for-loop-t}
Sample $S \subseteq V$ with sample-rate $p$\;
\label{line:sample-S}
\If{$|S| \leq 100 pn$}{
Run \Cref{lem:GR-forest} on $G[S]$ with set of known edges $E(G[S]) \cap K$ to find $F \subseteq E(G[S]) \setminus K$\;
\label{line:solve-forest}
$K \gets K \cup F$\;
}
}
\Return $K$\;
\end{algorithm}

Let $m_t$ be the number of undiscovered edges at the end of iteration $t$ in \Cref{alg:key-subroutine}. Initially, $m_0 = |E|$.

\begin{lemma}
\label[lemma]{lem:expected-undiscovered}
    For any $1 \le t \le \ell$, $\EX[m_t \mid m_{t-1}] \le (1-p^2/2) m_{t-1}$. 
\end{lemma}

\begin{proof}
We say the random set $S$ sampled at \Cref{line:sample-S} of \Cref{alg:key-subroutine} is \emph{good} if the following two events occur.
\begin{itemize}
    \item $\cE_1$: $|S| \leq 100pn$.
    \item $\cE_2$: $G[S]$ is a matching. 
\end{itemize}

\begin{claim} \label{clm:good-event} 
Let $(u, v) \in E \setminus K$ be an arbitrary unknown edge. If $u, v \in S$ and $S$ is good, then the edge $(u, v)$ is found by \Cref{lem:GR-forest} at \Cref{line:solve-forest}. 
\end{claim}

Hence, it suffices to bound the probability of the event above.

\begin{claim} \label{clm:edge-capture}
Let $(u, v) \in E \setminus K$ be an arbitrary unknown edge. Then, $\Pr[u,v \in S \text{ and } S \text{ is good}] \geq p^2/2$. \end{claim}

\begin{proof} First, clearly $\Pr[u,v \in S] = p^2$ and so it suffices to show that $\Pr[S \text{ is good} ~|~ u,v \in S] \geq 1/2$. We will establish this by showing that conditioned on $u,v \in S$, each of the events $\cE_1,\cE_2$ above fail to occur with probability at most $1/4$ and taking a union bound. 

Clearly, $\EX_S[|S| \mid u,v \in S] = p(n-2) + 2 = pn + 2(1-p)$. Since $n \ge n^{1/3} D^{2/3} \ge m^{1/3} D^{1/3}$ (we have $m \le nD$ as $D$ is an upper bound on the max-degree), we have $pn \ge 1/10$ and thus $pn + 2(1-p) \le pn + 2 \le 21 pn$. By Markov's inequality, $\prob[\neg \cE_1 \mid u, v \in S] \le \frac{21 pn}{100 pn} \le 1/4$. 

We now prove the desired bound for $\cE_2$. Note that $G[S]$ is \emph{not} a matching iff either (a) some vertex $z \in N(u) \cup N(v)$ appears in $S$, or (b) some vertex $w \in V \setminus (\{u, v\} \cup N(u) \cup N(v))$ appears in $S$ and has at least two neighbors appearing in $S$. Since $|N(u) \cup N(v)| \le 2D$, the probability of the former occurring is at most $2 p D = \frac{1}{5} \cdot \frac{D^{2/3}}{m^{1/3}} \le \frac{1}{5}$ (recall $D \le \sqrt{m}$). For the latter, fix $w \notin \{u, v\} \cup N(u) \cup N(v)$, and two of its neighbors $x, y \in N(w)$. The probability that $w, x, y$ all appear in $S$ is $p^3$. By union bound over all choices of $w, x, y$, the probability of (b) is bounded by 
\[
p^3 \cdot \left| \{(w, x, y): w \in V, x, y \in N(w)\}\right|  \le p^3 \cdot 2 m D \le \frac{1}{500}, 
\]
which completes the proof.
\end{proof}

By combining \cref{clm:good-event} and \cref{clm:edge-capture}, the probability that each undiscovered edge before the $t$-th iteration of the for loop at \cref{line:for-loop-t} becomes discovered after the iteration is at least $p^2 /2$. Therefore, $\EX[m_t \mid m_{t-1}] \le (1-p^2/2) m_{t-1}$ follows by linearity of expectation. 
\end{proof}

\begin{lemma}
    \label[lemma]{lem:expected-cost-iteration-t}
    For $1 \le t \le \ell$, the expected number of queries by \cref{alg:key-subroutine} in the $t$-th iteration of the loop at \cref{line:for-loop-t} is 
    \[
    O\left( \frac{\max(2,p^2 \EX[m_{t-1}])}{\log \max(2,p^2 \EX[m_{t-1}])} \cdot \log(100pn)\right).
    \]
\end{lemma}

\begin{proof}
    Clearly, the only way the algorithm makes queries is via \cref{lem:GR-forest}. Let $x_t$ be the number of unknown edges in $G[S]$ in the $t$-th iteration. The number of queries made by \cref{lem:GR-forest} is thus at most
    \[
    O\left( \frac{\max(2,x_t)}{\log \max(2,x_t)} \cdot \log(100pn)\right).
    \]
    The function $\frac{x_t}{\log x_t}$ is concave for sufficiently large $x_t$ ($x_t > e^2$). In order to apply Jensen's inequality, we define an auxiliary variable $y_t = e^2 + x_t$, and use 
    \[
    O\left( \frac{y_t \log(100 pn)}{\log y_t}\right)
    \]
    to upper bound the number of queries used in the iteration. By Jensen's inequality, 
    \[
    \EX\left[\frac{y_t \log(100 pn)}{\log y_t}\right] \le \frac{\EX[y_t] \log(100 pn)}{\log \EX[y_t]} \le \frac{(\EX[x_t]+e^2) \log(100 pn)}{\log (\EX[x_t]+e^2)} = \frac{(p^2 \EX[m_{t-1}] +e^2) \log(100 pn)}{\log (p^2 \EX[m_{t-1}]+e^2)},
    \]
    which is further upper bounded by 
    \[
    O\left( \frac{\max(2,p^2 \EX[m_{t-1}])}{\log \max(2,p^2 \EX[m_{t-1}])} \cdot \log(100pn)\right).
    \]
\end{proof}

\begin{lemma}
\label[lemma]{lem:expected-cost-key-subroutine}
    The expected number of queries of \cref{alg:key-subroutine} is 
    \[
    O\left(\left(\frac{|E|}{\log \max(2, p^2|E|)} + \ell\right)\cdot \log(100 pn)\right)
    \]
\end{lemma}

\begin{proof}
    By combining \cref{lem:expected-undiscovered} and \cref{lem:expected-cost-iteration-t}, the expected cost of the algorithm can be bounded by 
    \begin{align*}
        O\left( \sum_{t=0}^{\ell-1} \frac{\max(2,p^2 \EX[m_{t}])}{\log \max(2,p^2 \EX[m_{t}])} \cdot \log(100pn)\right) &\le O\left( \sum_{t=0}^{\ell-1} \frac{\max(2,p^2 (1-p^2/2)^t |E|))}{\log \max(2,p^2 (1-p^2/2)^t |E|)} \cdot \log(100pn)\right).
    \end{align*}
    We decompose the above sum based on values of $p^2 (1-p^2/2)^t |E|$. 
    \begin{itemize}
        \item $p^2 (1-p^2/2)^t |E| < 2$. We can bound the sum from these terms by $O(\ell \cdot \log(100 pn))$. 
        \item $p^2 (1-p^2/2)^t |E| \in [2^i, 2^{i+1})$ for some integer $1 \le i \le \log (p^2|E|)$. The number of such terms is $O(1/p^2)$, and each term contributes $O(\frac{2^i}{i} \cdot \log(100pn))$ to the total sum. Hence, in total, these terms contribute $O(\frac{1}{p^2} \cdot \frac{2^i}{i} \cdot \log(100pn))$. 
    \end{itemize}
    The total contribution of the second case above can be written as
    \[
    O\left(\sum_{i=1}^{\lfloor \log (p^2 |E|)\rfloor} \frac{1}{p^2} \cdot \frac{2^i}{i} \cdot \log(100 pn)\right). 
    \]
    It is a simple exercise to verify $\sum_{i=1}^{q} \frac{2^i}{i} = O(\frac{2^q}{q})$, so the above can be bounded by 
     \[
    O\left(\frac{|E|}{\log \max(2, p^2|E|)} \cdot \log(100 pn)\right). 
    \]
\end{proof}

\begin{proof}[Proof of \cref{lem:key-subroutine}]
The bound on the expected value of $|E \setminus K|$ follows by repeatedly applying \cref{lem:expected-undiscovered} $\ell$ times, and the expected cost bound follows from \cref{lem:expected-cost-key-subroutine}.  
\end{proof}

The following two corollaries follow by applying \cref{lem:key-subroutine} in two different ways. 

\begin{corollary}
\label[corollary]{cor:key-subroutine-variant1}
There is an algorithm that, given given $\CC$-query access to an underlying (unknown) $n$-node graph $G = (V, E)$, the vertex set $V$, an upper bound $m$ on the number of edges, an upper bound $D \le \sqrt{m}$ on the maximum degree, and a parameter $0 < \delta < 1$, outputs the graph with probability $\ge 1-\delta$ using \[
    O\left(\left(\frac{|E|}{\log \max(2, p^2|E|)} + \frac{\log m + \log \frac{1}{\delta}}{p^2}\right)\cdot \log(100 pn)\right)
    \] $\CC$-queries in expectation for $p = \frac{1}{10m^{1/3}D^{1/3}}$. 
\end{corollary}
\begin{proof}
    We apply \cref{lem:key-subroutine} with parameter $\ell = \lceil \frac{2}{p^2}(\ln m + \ln \frac{1}{\delta}) \rceil$. The expected cost easily follows. To bound the error probability, note that 
    \begin{align*}
        \EX[|E \setminus K|] & \le (1-p^2/2)^\ell |E| \le (1-p^2/2)^{2/p^2(\ln m + \ln \frac{1}{\delta})} \cdot m = \frac{\delta}{m} \cdot m = \delta. 
    \end{align*}
    Hence, by Markov's inequality, $\prob[|E \setminus K| \ge 1] \le \delta$, so $\prob[E = K]$ (i.e., the algorithm finds all edges) is at least $1-\delta$. 
\end{proof}

\begin{corollary}
\label[corollary]{cor:key-subroutine-variant2}
There is an algorithm that, given given $\CC$-query access to an underlying (unknown) $n$-node graph $G = (V, E)$, the vertex set $V$, an upper bound $m$ on the number of edges, and an upper bound $D \le \sqrt{m}$ on the maximum degree of $G$, outputs $G$ using \[
    O\left(\left(\frac{|E|}{\log \max(2, p^2|E|)} + \frac{\log (n^2 / m)}{p^2}\right)\cdot \log(100 pn) + (n+m^2/n^2) \cdot \log \frac{n^2}{n + m^2/n^2}\right)
    \] $\CC$-queries in expectation, where $p = \frac{1}{10m^{1/3}D^{1/3}}$. 
\end{corollary}

\begin{proof}
We first apply \cref{lem:key-subroutine} with parameter $\ell = \lceil \frac{2 \ln(n^2 / m)}{p^2}\rceil$. The expected number of undiscovered edge $m_u$ afterwards is 
\[
O(m \cdot (1-p^2/2)^{\ell}) = O(m^2 / n^2),
\]
and the expected cost is 
\[
O\left(\left(\frac{|E|}{\log \max(2, p^2|E|)} + \frac{\log (n^2 / m)}{p^2}\right)\cdot \log(100 pn)\right).
\]
Afterwards, we run \cref{lem:GR-binary-search} to reconstruct the whole graph. The number of queries used by \cref{lem:GR-binary-search} is 
\[
O((n+m_u) \cdot \log \frac{n^2}{n + m_u}). 
\]
By Jensen's inequality, the expected number of queries is 
\[
O\left((n+\EX[m_u]) \cdot \log \frac{n^2}{n + \EX[m_u]}\right) = O\left((n+m^2/n^2) \cdot \log \frac{n^2}{n + m^2/n^2}\right). 
\]

\end{proof}

\subsection{Details of the Adaptive Algorithm} \label{sec:a-alg-proof}

In this section, we prove our \cref{thm:adaptive-alg}, which we recall below: 

\AdaptiveAlgo*

\begin{proof} First, we define a sequence $\alpha$ where 
\[
\alpha_1 = m^{1/3} \text{    and    } \alpha_i = \alpha_{i-1}^{0.9} \text{ for } i > 1.
\]
We stop at $\alpha_\ell \le 100$. Clearly, $\ell = O(\log \log m)$. 

Then we define a sequence
\[T_i = \sqrt{m / \alpha_i} \text{ for } i \ge 1.\]

Next, for every $1 \le i \le \ell$, we run \Cref{lem:find-high-degree-nodes} with parameters $V, m, T = T_i$ and $\delta = 1/10\ell$. We let the output of \Cref{lem:find-high-degree-nodes} be $H_i$. By a union bound, with probability $0.9$, all vertices in $H_i$ have degree more than $T_i/2$, and all vertices with degree at least $2T_i$ are in $H_i$. In the following, we assume these degree bounds hold. The query complexity of these calls of \Cref{lem:find-high-degree-nodes} is 
\[
O\left(\sum_{1 \le i \le \ell} \frac{m(\log m + \log(1/\delta)) \log n}{T_i}\right) \le O\left( \ell \cdot \frac{m(\log m + \log(\ell)) \log n}{m^{1/3}}\right) \leq O\left(\frac{m \log n}{\log m}\right). 
\]
We additionally set $H_0 = V$. Finally, for every $1 \le i \le \ell$, define $S_i = H_{i-1} \setminus H_{i}$. Note that if the guarantees of \Cref{lem:find-high-degree-nodes} hold, then the maximum degree of vertices in $S_i$ are $\le 2T_i$ for $i \ge 1$, and the minimum degree of vertices in $S_i$ are $\ge T_{i-1}/2$ for $i > 1$.

Then, our algorithm reconstructs edges in two cases, depending on whether the two endpoints of the edge are from two sets $S_i$, $S_j$ with $|i - j| \le 1$ or not. 

\paragraph{Reconstruct edges inside $S_i \cup S_{i+1}$. } Here, we further consider two different cases, depending on whether $i = 1$ or not. 
\begin{itemize}
    \item Reconstruct edges inside $S_1 \cup S_{2}$. Here, we run \cref{cor:key-subroutine-variant1} with parameters $V = S_1 \cup S_2, m, D = 2T_2 = \Theta(m^{0.35})$ and $\delta = 1/100$. The success rate is $1-\delta = 0.99$, and the expected query complexity is 
    \[
    O\left(\left(\frac{|E|}{\log \max(2, p^2|E|)} + \frac{\log m + \log \frac{1}{\delta}}{p^2}\right)\cdot \log(100 pn)\right)
    \]
    for $p = \Theta((mD)^{-1/3}) = \Theta(m^{-0.45})$. The query complexity simplifies to 
    \[
    O\left( \frac{m}{\log (m^{0.1})} + \frac{\log m}{m^{-0.9}} \cdot \log(100 pn)\right) = O\left(\frac{m \log n}{\log m} \right). 
    \]
    \item Reconstruct edges inside $S_i \cup S_{i+1}$ for $i \ge 2$. In this case, we run \Cref{cor:key-subroutine-variant2} with parameters $V = S_i \cup S_{i+1}, m, D = 2T_{i+1}$. Note that as the minimum degree of vertices in $S_i \cup S_{i+1}$ is $\Omega(T_{i-1})$, $|S_i \cup S_{i+1}| = O(m/T_{i-1})$. Let $e_{i, i+1}$ denote the number of edges in $G[S_i \cup S_{i+1}]$. Then, the expected number of queries used by \Cref{cor:key-subroutine-variant2} is %
    \begin{align*}
     O&\left(\left(\frac{e_{i,i+1}}{\log \max(2, p^2 \cdot e_{i,i+1})} + \frac{\log ((m/T_{i-1})^2 / m)}{p^2}\right)\cdot \log(100 p \cdot (m/T_{i-1})) \right.\\
    & \left. + (m/T_{i-1}+m^2/(m/T_{i-1})^2) \cdot \log\left( \frac{(m/T_{i-1})^2}{m/T_{i-1} + m^2/(m/T_{i-1})^2}\right)\right)\\
    = O &\left(\left(\frac{e_{i,i+1}}{\log \max(2, p^2 \cdot e_{i,i+1})} + \frac{\log (m/T_{i-1}^2)}{p^2}\right)\cdot \log(100 p m/T_{i-1}) + T_{i-1}^2 \cdot \log\left(\frac{m}{T_{i-1}^2}\right)\right),
    \end{align*}
    where $p = \Theta((mT_{i+1})^{-1/3})$. Note that to bound the second term, we used the fact that $T_{i-1} = \Omega(m^{1/3})$ and so $T_{i-1}^2 = \Omega(m/T_{i-1})$. By plugging in the value of $p$ and $T_{i-1} = \sqrt{m / \alpha_{i-1}}$, $T_{i+1} = \sqrt{m / \alpha_{i+1}}$, the above simplifies to 
     \begin{align*}
    & O\left(\left(\frac{e_{i,i+1}}{\log \max(2, \frac{\alpha_{i+1}^{1/3}}{m} \cdot e_{i,i+1})} + \frac{m \log (\alpha_{i-1})}{\alpha_{i+1}^{1/3}}\right)\cdot \log(\alpha_{i+1}^{1/6} \alpha_{i-1}^{1/2}) + m \cdot \frac{\log(\alpha_{i-1})}{\alpha_{i-1}}\right) \\
    = & O\left(\left(\frac{e_{i,i+1}}{\log \max(2, \frac{\alpha_{i+1}^{1/3}}{m} \cdot e_{i,i+1})}\right) \cdot \log(\alpha_{i+1}) + m \cdot \left(\frac{\log^2(\alpha_{i+1})}{\alpha_{i+1}^{1/3}} + \frac{\log(\alpha_{i-1})}{\alpha_{i-1}} \right) \right). 
    \end{align*}
    When $e_{i, i+1} \le m / \alpha_{i+1}^{0.1}$, the first term can be bounded by $O( (m / \alpha_{i+1}^{0.1}) \cdot \log(\alpha_{i+1})) = O(m / \alpha_{i+1}^{0.01})$. Otherwise, the first term can be bounded by $O(e_{i, i+1})$. Therefore, the above can be upper bounded by 
    \[
    O\left(e_{i, i+1} + m \cdot \left(\frac{1}{\alpha_{i+1}^{0.01}}+\frac{\log^2(\alpha_{i+1})}{\alpha_{i+1}^{1/3}} + \frac{\log(\alpha_{i-1})}{\alpha_{i-1}} \right)\right) = O\left(e_{i, i+1} + \frac{m}{\alpha_{i+1}^{0.01}}\right). 
    \]
    Summing over all $i \ge 2$, the total expected query complexity can hence be bounded by
    \[
    \sum_{i = 2}^{\ell-1} e_{i,i+1} + m \sum_{i=3}^{\ell} \alpha_{i}^{-0.01} = O(m) = O\left(\frac{m\log n}{\log m}\right)
    \]
    where we have used the fact that $\sum_{i=3}^{\ell} \alpha_{i}^{-0.01} = \Theta(1)$.
\end{itemize}
\paragraph{Reconstruct edges between $S_i$ and $S_j$ for $1 \le i \le j - 2$. } Recall the maximum degree of vertices in $S_i$ is $2T_i$ and we have computed all edges in $G[S_i]$ by the previous case. Hence, we can apply \cref{lem:GR-find-neighbors-IS} with $V = S_i \cup S_j, U = S_i, D = 2T_i$ to compute all adjacent edges between any given $v \in S_j$ and $S_i$ using 
\[
O\left((T_i + \deg(v, S_i)) \cdot  \frac{\log (|S_i| / T_i + 1)}{ \log(\max(2, \deg(v, S_i) / T_i))}\right)
\]
queries. Here, we also need to consider two cases, depending on whether $i=1$ or not. 
\begin{itemize}
    \item $i = 1$. In this case, we simply use $n$ to upper bound $|S_i|$, and the query complexity simplifies to 
    \[
    O\left((T_1 + \deg(v, S_1)) \cdot  \frac{\log  n}{ \log(\max(2, \deg(v, S_1) / T_1))}\right). 
    \]
    We can further use 
    \[
    O\left(T_1 \log n + \max(\deg(v, S_1), T_2) \cdot  \frac{\log  n}{ \log(T_2 / T_1))}\right)
    \]
    to upper bound the above since the second term in the sum is an increasing function of $\deg(v,S_1)$. The above is now upper bounded by
    \[
    O\left((T_2 + \deg(v, S_1)) \cdot \frac{\log n}{\log m}\right),
    \]
    since $T_2/T_1 = \Theta(m^c)$ for a constant $c > 0$. Next, we sum the above bound over all $v \in \cup_{j = 3}^{\ell} S_j$. Note that for every such $v$, $\deg(v) \ge \Omega(T_2)$, and so the number of such $v$ is bounded by $|\cup_{j = 3}^{\ell} S_j| \leq O(m / T_2)$. Hence, the sum of the above bound over all $v \in \cup_{j = 3}^{\ell} S_j$ is bounded by 
    \[
    \left( \sum_{v \in \cup_{j=3}^{\ell} S_j} (T_2 + \deg(v,S_1)) \right) \cdot O\left(\frac{\log n}{\log m}\right) \leq O\left(\frac{m \log n}{\log m}\right). 
    \]
    \item $i \ge 2$. In this case, we use $O(m / T_{i-1})$ to bound $|S_i|$, as all vertices in $S_i$ have degree $\Omega(T_{i-1})$. Hence, the query complexity for a fixed $v$ simplifies to 
    \[
    O\left((T_i + \deg(v, S_i)) \cdot  \frac{\log ((m / T_{i-1}) / T_i)}{ \log(\max(2, \deg(v, S_i) / T_i))}\right). 
    \]
We can further use 
\begin{align*}
& O\left(T_i \cdot  \log ((m / T_{i-1}) / T_i) + \max(T_{i+1} / \alpha_i^{0.01}, \deg(v, S_i)) \cdot  \frac{\log ((m / T_{i-1}) / T_i)}{ \log((T_{i+1}/ \alpha_i^{0.01}) / T_i))}\right)
\end{align*}
to upper bound the above (by considering whether $\deg(v, S_i) \le T_{i+1}/\alpha_{i}^{0.01}$ or not). This further simplifies to 
\begin{align*}
&O\left(T_i \cdot \log(\alpha_{i-1}^{1/2} \alpha_i^{1/2}) + (T_{i+1} /\alpha_i^{0.01} + \deg(v, S_i)) \cdot \frac{\log(\alpha_{i-1}^{1/2} \alpha_i^{1/2})}{\log(\alpha_i^{1/2-0.01} / \alpha_{i+1}^{1/2})}\right) \\
&= O\big(T_{i+1} /\alpha_i^{0.01} + \deg(v, S_i)\big) \text{.} 
\end{align*}
since $T_i = T_{i+1}/\alpha_i^{0.05}$, implying that $T_i \cdot \log(\alpha_{i-1}^{1/2} \alpha_i^{1/2}) = o(T_{i+1}/\alpha_{i}^{0.01})$. Finally, note that every vertex in $\cup_{j=i+2}^\ell S_j$ has degree $\Omega(T_{i+1})$ and so $|\cup_{j=i+2}^\ell S_j| = O(m/T_{i+1})$. The total query complexity over all $i \ge 2$ and all $v \in S_j$ for $j \ge i + 2$ is thus
\begin{align*}
    & O\left(\sum_{2 \le i \le \ell-2} \sum_{v \in \bigcup_{j=i+2}^\ell S_j} \left(T_{i+1} /\alpha_i^{0.01} + \deg(v, S_i)\right)\right) \\
    &= O\left(\sum_{2 \le i \le \ell-2} \left| \cup_{j=i+2}^\ell S_j\right| \cdot T_{i+1} /\alpha_i^{0.01} + \sum_{2 \le i \le \ell-2} \sum_{v \in \bigcup_{j=i+2}^\ell S_j} \deg(v, S_i)\right)\\
    &= O\left(\sum_{2 \le i \le \ell-2} m/\alpha_i^{0.01} + m\right) = O(m). 
\end{align*}
\end{itemize}

Summing over each stage of the algorithm, we observe that the total expected number of queries is bounded by $O(\frac{m\log n}{\log m})$. So far, the algorithm has constant success probability. As we can verify whether the graph we reconstruct is correct by verifying all discovered edges are indeed edges and the total number of discovered edges is $m$, we can repeat the algorithm until we find the correct graph. The expected number of repeats is $O(1)$ and hence the expected query complexity is $O(\frac{m \log m}{\log n})$. \end{proof}

\section{Algorithm Using Two Rounds of Adaptivity} \label{sec:2-round-proof}

In this section we obtain a simple algorithm using two rounds of adaptivity making $O(m \log n + n \log^2 n)$ queries which successfully reconstructs the underlying graph with probability $1-1/\poly(n)$. This is in contrast with non-adaptive algorithms (one round) which require $\Omega(n^2)$ queries, even for $m = O(n)$ (recall \Cref{thm:LB-NA}). 

\TwoRound*

The first round of queries is used to attain a constant factor approximation of the degree of every vertex. This is accomplished using the following lemma.

\begin{lemma} \label{lem:NA-degree} Fix a vertex $u \in V$. There is a non-adaptive algorithm, \cref{alg:NA-degree}, which makes $O(\log n \cdot \log(1/\eps))$ $\CC$-queries and returns a number $\widetilde{d}_u$ for which  $\widetilde{d}_u \in [\deg_G(u),4\deg_G(u)]$ with probability $1-\eps$. \end{lemma}
\begin{proof} Our estimator works by checking how likely it is for a random set of a given size to intersect the neighborhood of $u$. We first observe that this event can be checked with two $\CC$-queries. 

\begin{fact} \label{fact:neighbor-check} Given $u \in V$ and $S \subseteq V \setminus u$ we have
\[
\CC(S \cup \{u\}) - \CC(S) = \boldsymbol{1}(\deg(u,S) = 0)\text{,}
\]
i.e. two $\CC$-queries suffice to check if $S$ contains a neighbor of $u$. \end{fact}

\begin{proof} The fact follows by observing that $u$ has a neighbor in $S$ if and only if adding $u$ to $S$ does not increase the number of connected components. \end{proof}

\begin{algorithm}[ht]
\caption{Non-adaptive degree estimator} \label{alg:NA-degree}
\KwIn{ $n \in \mathbb{N},\eps \in (0,1)$, $\CC$-query access to a graph $G=(V, E)$ with $n$ nodes, and $u \in V$. }
\KwOut{$\widetilde{d}_u \in \mathbb{N}$ satisfying $\widetilde{d}_u \in [\deg_G(u),4\deg_G(u)]$ with probability $1-\eps$. }
 Let $\ell = \lceil 450 \ln (800/\eps^2) \rceil$\; 
\For{$p = 1,2, \ldots, \ceil{\log (n-1)}$}{
\For{$j = 1,\ldots,\ell$}{
\label{line:alg:NA-degree:inner-for-loop-start}
Draw random set $S$ containing $2^p$ iid uniform samples from $V\setminus\{u\}$, and compute the Boolean random variable $X^{(p)}_j := \boldsymbol{1}(\deg(u,S) = 0)$ by making two $\CC$-queries (\Cref{fact:neighbor-check})\;
Let $Z_p := \frac{1}{\ell}\sum_{j \leq \ell} X_j^{(p)}$\;
}
\label{line:alg:NA-degree:inner-for-loop-end}
}
Let $p^{\star}$ be the largest $p \in [\log n]$ for which $Z_p > \frac{1}{2e}$\;
\label{line:alg:NA-degree:p-star}
\Return $\widetilde{d}_u = \frac{2(n-1)}{2^{p^{\star}}}$\;
\label{line:alg:NA-degree:output}
\end{algorithm}

The algorithm is described in \cref{alg:NA-degree}, and we now prove its correctness. The inner for-loop from \cref{line:alg:NA-degree:inner-for-loop-start} to \cref{line:alg:NA-degree:inner-for-loop-end} computes an empirical estimate for the probability of a random set of $2^p$ vertices being disjoint from the neighborhood of $u$. This probability is exactly
\begin{align} \label{eq:degree-est}
    \prob_{S \colon |S| = 2^p} \left[ \deg(u,S) = 0\right] = \left(1-\frac{\deg_G(u)}{n-1}\right)^{2^p}.
\end{align}

We will argue that with probability $1-\eps$, the value $p^{\star}$ computed in \cref{line:alg:NA-degree:p-star} satisfies $2^{p^\star} \in (\frac{n-1}{2\deg_G(u)},\frac{2(n-1)}{\deg_G(u)}]$, and therefore the output in \cref{line:alg:NA-degree:output} satisfies $\widetilde{d}_u \in [\deg_G(u),4\deg_G(u)]$ as desired. This claim holds due to the following \Cref{clm:Z}, which completes the proof. \end{proof}

\begin{claim} \label{clm:Z} Consider the random variables $Z_p$ and $X^{(p)}_j$ defined in \cref{alg:NA-degree}. With probability $1-\eps$ the following hold.
\begin{itemize}
    \item For $p$ such that $2^{p} \in (\frac{n-1}{2\deg_G(u)},\frac{n-1}{\deg_G(u)}]$, we have $Z_p > \frac{1}{2e}$. 
    \item For every $p$ such that $2^p > \frac{2(n-1)}{\deg_G(u)}$ we have $Z_p < \frac{1}{2e}$.
\end{itemize}
\end{claim}

\begin{proof} Define $\mu_p$ to be the quantity from \cref{eq:degree-est} and observe that $\EX[Z_p] = \EX[X_j^{(p)}] = \mu_p$. We show that the two events described in the statement of the claim each fail to hold with probability at most $\eps/2$ and the claim then follows by a union bound. 

Consider the first bullet point. Let $p$ be such that $2^{p} \in (\frac{n-1}{2\deg_G(u)},\frac{n-1}{\deg_G(u)}]$ and observe that $\mu_p \geq 1/e$. Recall that $Z_p = \frac{1}{\ell} \sum_{j \leq \ell} X^{(p)}_j$ and $X^{(p)}_j = \boldsymbol{1}(\deg(u,S) = 0)$. We have
\begin{align}
    \prob_{Z_p}[Z_p \leq 1/2e] \leq \prob_{Z_p}[|Z_p - \mu_p| > 1/2e] \leq 2 \exp\left(-2 \cdot\frac{\ell}{4e^2}\right) \leq \eps/2
\end{align}
where the first inequality follows from our lower bound on $\mu_p$, the second inequality is due to Hoeffding's bound, and the third inequality is because $\ell \geq 2e^2\ln(4/\eps)$. %

We now handle the second bullet point. We first handle $p$ for which $2^p$ is significantly larger than the target interval. If $2^p \geq t \cdot \frac{n-1}{\deg_G(u)}$, then $\mu_p \leq \exp(-t)$ and so by Markov's inequality $\prob[Z_p \geq 1/2e] \leq 2e \cdot \exp(-t)$. Therefore, by a union bound over all $p$ such that $2^p > \ln(100/\eps) \cdot \frac{n-1}{\deg_G(u)}$, we have
\begin{align}
    \prob\left[\exists p \text{ for which } 2^p > \ln(100/\eps) \cdot \frac{n-1}{\deg_G(u)} \text{ and } Z_p \geq 1/2e\right] \leq 2e \sum_{q \colon 2^q > \ln(100/\eps)} \exp(-2^q) < \eps/4 \text{.}
\end{align}
Next we handle values of $p$ such that $\frac{2(n-1)}{\deg_G(u)} < 2^p \leq \ln(100/\eps) \cdot \frac{n-1}{\deg_G(u)}$, of which there are only $\log (\ln (100/\eps))$. For such a value we have $\mu_p < 1/e^2$ and so using Hoeffding's bound we have
\begin{align}
    \prob[Z_p \geq 1/2e] \leq  \prob[|Z_p - \mu_p| > 1/30] \leq 2\exp\left(-2 \cdot \frac{\ell}{900}\right) \leq \frac{\eps^2}{400} \leq  \frac{\eps}{4\log (\ln (100/\eps))},
\end{align}
where the first inequality used the upper bound on $\mu_p$, the second inequality used Hoeffding's bound, and the second to last inequality used $\ell > 450 \ln (800/\eps^2)$. The final inequality is simply due to $4\log \ln(100/\eps) < \frac{400}{\eps}$ for all $\eps > 0$. %
Taking a union bound over all $p$ satisfying $\frac{2(n-1)}{\deg_G(u)} < 2^p \leq \ln(100/\eps) \cdot \frac{n-1}{\deg_G(u)}$ completes the proof.
\end{proof}

The second round of queries is used to learn the neighborhood of each vertex using the degree approximation obtained in the previous round. This is accomplished using the following lemma which we prove using a basic group testing primitive. 

\begin{lemma} \label{lem:NA-nbhd} Fix a vertex $u \in V$ and a value $d \geq \deg_G(u)$. There is a non-adaptive algorithm using $O(d \log (n/\alpha))$ $\CC$-queries which returns a set $T \subseteq V$ such that $T = N_G(u)$ with probability $1-\alpha$. Additionally, if $d < \deg_G(u)$, then the algorithm outputs $\mathsf{fail}$ with probability $1$. \end{lemma}

\begin{proof} Consider the Boolean vector $x^u \in \{0,1\}^n$ for which $x^u_j = 1$ iff the $j$-th vertex in $V$ is a neighbor of $u$ (under an arbitrary labeling of the vertices). We reduce to the problem to recovering the support of $x^u$ using $\mathsf{OR}$ queries. Given $S \subseteq V \setminus \{u\}$, such a query is defined as
\[
\mathsf{OR}_S(x^u) = \bigvee_{j \in S} x^u_j
\]
By \Cref{fact:neighbor-check}, we can use two $\CC$-queries to the graph to simulate one $\mathsf{OR}$-query to the vector $x_u$. In particular, given $S \subseteq V \setminus u$, we have
\[
\CC(S \cup \{u\}) - \CC(S) = \boldsymbol{1}(\deg(u,S) = 0) = \mathsf{OR}_S(x^u) \text{.}
\]
\Cref{lem:NA-nbhd} is now an immediate corollary of the following standard group testing lemma, whose proof can be found, for example, in \cite{BLMS24}. \end{proof}

\begin{lemma} [\cite{BLMS24}, Lemma 1.9] \label{lem:t-ER} Let $v \in \{0,1\}^n$ and let $\textsf{supp}(v) = \{j \in [n] \colon x_j = 1\}$. Given $n, d, \alpha$, there is a non-adaptive algorithm that makes $O(d \log (n/\alpha))$ $\mathsf{OR}$ queries and if $|\textsf{supp}(v)| \leq d$, returns $\textsf{supp}(v)$ with probability $1-\alpha$, and otherwise certifies that $|\textsf{supp}(v)| > d$ with probability $1$. 
\end{lemma}

\begin{proof} The algorithm is defined using the above two lemmas as follows.

\begin{itemize}
    \item (\nth{1} Round) For every $u \in V$, run the procedure of \Cref{lem:NA-degree} using error probability $\eps = 1/n^2$, and let $\widetilde{d}_u$ be the corresponding output. If $\sum_{u \in V} \widetilde{d}_u > 8m$, then return $\mathsf{fail}$. 
    \item (\nth{2} Round) For every $u \in V$, run the procedure of \Cref{lem:NA-nbhd} using $d = \widetilde{d}_u$ and error probability $\alpha = 1/n^2$. Use the output as the neighborhood of $u$. 
\end{itemize}

The first round uses $O(n \log^2 n)$ queries and outputs $\widetilde{d}_u \in [\deg_G(u),4\deg_G(u)]$ for every $u \in V$ with probability at least $1-\eps n = 1-1/n$ by a union bound. Conditioned on this, we have 
\begin{align} \label{eq:edge-bound}
    \sum_{u \in V} \widetilde{d}_u \leq 4 \sum_{u \in V} \deg_G(u) = 8|E(G)| \leq 8m
\end{align}
and so the algorithm continues to the second round, which successfully returns the neighborhood of every vertex with probability at least $1-\alpha n = 1-1/n$ again by a union bound. Therefore, the algorithm succeeds with probability at least $1-2/n$ by a union bound over the two rounds. By \Cref{lem:NA-nbhd} and \cref{eq:edge-bound}, the number of queries made in the second round is bounded as %
\[
\sum_{u \in V} O(\widetilde{d}_u \log n) \leq O(m \log n)
\]
in the worst case, and this completes the proof. \end{proof}

\section{Lower Bounds} \label{sec:LB}

\subsection{Non-Adaptive Lower Bound} \label{sec:NA-LB}

In this section we prove \Cref{thm:LB-NA}, establishing a $\CC$-query complexity lower bound of $\Omega(n^2)$ for non-adaptive algorithms, even when the graph is known to have only $O(n)$ edges.

\NonAdaptiveLower* 

\begin{proof} For every distinct pair of vertices $\{u,v\} \in \binom{V}{2}$, we define a pair of graphs $G^0_{u,v}(V,E_0)$ and $G^1_{u,v}(V,E_1)$ as follows. The edges of $G^0_{u,v}$ are given by $E_0 := \{(u,w), (w,v) \colon w \in V \setminus \{u,v\}\}$, i.e. $G^0_{u,v}$ is a union of $2$-paths joining $u$ and $v$, through every other vertex $w$. The edges of $G^1_{u,v}$ are simply defined as $E_1 := E_0 \cup \{(u,v)\}$. Observe that for an algorithm to distinguish these two graphs it clearly must query a set $S$ that contains both $u$ and $v$. However, if $S$ contains any additional vertices besides $u,v$, then the subgraph induced on $S$ is connected for both graphs, and so $\CC(S)$ returns $1$ for both $G^0_{u,v}$ and $G^1_{u,v}$. Therefore, if a set of queries $Q \subseteq 2^V$ distinguishes these two graphs then it must be the case that $\{u,v\} \in Q$. Moreover, this implies that the number of pairs $\{u,v\} \in \binom{V}{2}$ such that $Q$ distinguishes $G^0_{u,v}$ and $G^1_{u,v}$ is at most $|Q|$. 

Now, let $A$ be any non-adaptive $\CC$-query algorithm which successfully recovers any arbitrary input graph with probability $\geq 1/2$. The algorithm $A$ queries a random set $Q \subseteq 2^V$ according to some distribution, $\cD_A$. In particular, for every $\{u,v\} \in \binom{V}{2}$, $Q$ distinguishes $G^0_{u,v}$ and $G^1_{u,v}$ with probability $\geq 1/2$. Thus,
\begin{align}
    \frac{1}{2} \binom{n}{2} &\leq \sum_{\{u,v\} \in \binom{V}{2}} \prob_{Q \sim \cD_A}[Q \text{ distinguishes } G^0_{u,v} \text{ and } G^1_{u,v}] \nonumber \\
    &= \EX_{Q \sim \cD_A}\left[ \left|\left\{\{u,v\} \in \binom{V}{2} \colon Q \text{ distinguishes } G^0_{u,v} \text{ and } G^1_{u,v} \right\}\right| \right] \leq \EX[|Q|] \nonumber
\end{align}
using linearity of expectation, and this completes the proof. \end{proof}

\subsection{Adaptive Lower Bound} \label{sec:A-LB}

In this section we prove \Cref{thm:LB-A}, establishing a $\CC$-query complexity lower bound of $\Omega(\frac{m \log n}{\log m})$ for arbitrary fully adaptive algorithms, for all values of $m$, and even when $m$ is provided to the algorithm.

\AdaptiveLower*

\begin{proof} We first consider the case of relatively small $m$ and give a simple information-theoretic argument. Concretely, let us assume $m \leq n$. Observe that for any $U \subseteq V$, the number of connected components in $G[U]$ is at least $|U|-m$ and at most $|U|$. Thus, any $\CC$-query returns at most $\log (m+1)$ bits of information. Since $m \leq n$, the total number of graphs with $n$ nodes and $m$ edges is 
\[
\binom{\binom{n}{2}}{m} \geq \left(\binom{n}{2} / m \right)^m \geq \left(\binom{n}{2} / n \right)^m = \left(\frac{n-1}{2}\right)^m = 2^{\Omega(m \log n)}
\]
implying that $\Omega(\frac{m \log n}{\log m})$ $\CC$-queries are needed in this case. \smallskip

Now suppose $m \geq n$. Note that in this case $\frac{m \log n}{\log m} = \Theta(m)$ and so it suffices to prove an $\Omega(m)$ lower bound. Consider the largest integer $n_0$ for which $\binom{n_0}{2} - 1 \leq m$. Let $m_{\text{dum}} = m - (\binom{n_0}{2} - 1 )$ and observe that 
\begin{align} \label{eq:dummy}
    m_{\text{dum}} < \left(\binom{n_0 + 1}{2} - 1 \right) - \left(\binom{n_0}{2} - 1 \right) = n_0 \text{.}
\end{align}
We define a family $\cG$ of $n$-node graphs with $m$ edges as follows. Fix a set $K$ of $n_0$ vertices and a set $I$ of $n-n_0$ vertices. Each graph $G \in \cG$ is defined as follows. First, make $G[K]$ an $n_0$-clique, but choose a single edge and remove it. This accounts for $\binom{n_0}{2} - 1$ edges. If $n = n_0$, i.e. $m = \binom{n}{2} - 1$, then $I = \emptyset$ and this completes the construction. Otherwise, we complete the construction by taking a single vertex $z_{\text{dum}} \in I$ and adding a dummy edge from $z_{\text{dum}}$ to $m_{\text{dum}}$-many vertices in $K$. (This is possible by \cref{eq:dummy}). Note the total number of edges is now $m$. 

Next, we reduce the task of finding the unique zero in a Boolean array of length $\binom{n_0}{2} = m + 1 - m_{\text{dum}}$ (unstructured search) to the task of reconstructing an arbitrary graph $G \in \cG$ using $\CC$-queries. As unstructured search requires $\Omega(m + 1 - m_{\text{dum}}) = \Omega(m)$ queries to achieve  success probability $1/2$, graph reconstruction will also require $\Omega(m)$ queries. (By \cref{eq:dummy}, $m_{\text{dum}} < n_0 = O(\sqrt{m})$.) 

We can arbitrarily map the elements in the Boolean array to potential edges in $G[K]$, where a one corresponds to an edge, and a zero corresponds to a non-edge. Suppose we have an algorithm $\cA$ for reconstructing $G \in \cG$, i.e., the algorithm is able to discover where the non-edge in $K$ is. Then we construct an algorithm $\cA'$ for finding the unique zero in the Boolean array by simulating $\cA$. Whenever $\cA$ makes a query to a subset $S$ where $|S \cap K| \ne 2$, $\cA'$ does not need to make any query. Indeed, in this case the vertices in $S \cap K$ are always connected, and all edges in $(S \cap K) \times (S \setminus K)$ are fixed, so $\CC(S)$ is fixed for graphs in $\cG$. If $\cA$ makes a query to a subset $S$ where $S \cap K = \{u, v\}$, $\cA'$ first makes a query to the element corresponding to the edge $\{u, v\}$. All other edges among $S \cap K$ are fixed for graphs in $\cG$, so the value $\CC(S)$ can be computed using $1$ query. Whenever $\cA$ terminates and finds the non-edge in $K$, $\cA'$ also terminates and claims the unique zero is the element corresponding to that non-edge. Since the number of queries $\cA'$ makes never exceeds the number of queries $\cA$ makes, this concludes the reduction and thus the proof. \end{proof}

\section{Conclusion and Open Questions}

We have shown tight bounds on the $\CC$-query complexity of graph reconstruction. Our algorithm is adaptive, and we have shown that non-adaptive algorithms cannot do better than the $O(n^2)$ trivial upper bound even for graphs with $O(n)$ edges. This is in contrast with the graph reconstruction problem using $\cut$ and $\density$ queries, where there exist non-adaptive algorithms attaining the optimal query complexity. We believe it is an interesting direction to investigate the minimal number of rounds that suffice to attain the optimal $\Theta(\frac{m \log n}{\log m})$ $\CC$-query complexity. Towards this, we have obtained a two round algorithm using $O(m \log n + n \log^2 n)$ queries.

\begin{question} [Rounds of adaptivity] What is the minimal number of rounds of adaptivity sufficient to obtain a $O(\frac{m \log n}{\log m})$ $\CC$-query algorithm? As a starting point, what is the optimal query complexity for algorithms using two rounds of adaptivity? \end{question}

Our results show that $\CC$-queries and $\add/\cut$-queries have different power for the graph reconstruction problem in terms of rounds of adaptivity. It would be interesting to investigate to what extent these queries are comparable. One concrete question towards this is the following.

\begin{question} [Counting edges] How many $\CC$-queries suffice to learn, or approximate, the number of edges in an arbitrary input graph? \end{question}

In general, it will be interesting to quantify the number of $\CC$-queries required to derive or verify (testing) some specific property of the graph. Potentially such tasks can be performed with much fewer  queries than required for reconstructing the graph.

\bibliographystyle{alpha}
\bibliography{biblio}

\end{document}